\documentclass[conference]{ieee}
\usepackage{cite,amsmath,amssymb,amsthm,graphicx,subfigure,xspace,mathrsfs,enumitem,times,url}

\newcommand{\MVA}{{\sc MVA}\xspace}

\newcommand{\MoM}{{\sc MoM}\xspace}

\newtheorem{theorem}{Theorem}

\newtheorem{corollary}{Corollary}

\begin{document}
\title{The Multi-Branched Method of Moments for Queueing Networks}
%\title{The Multi-Branched Method of Moments for Queueing Networks}
%\title{Solving Multiclass Closed Queueing Networks by the Multi-Branched Method of Moments}
%\title{A Generalized Method of Moments for Multiclass Performance Models}
\author{Giuliano~Casale\\%[-5pt]
\em SAP Research\\%[-5pt]
\em TEIC Building, Shore Road\\
Newtownabbey, BT37 0QB, UK \\%[-5pt]
\em giuliano.casale@ieee.com
}
\date{}
\maketitle
\begin{abstract}
We propose a new exact solution algorithm for closed multiclass product-form queueing networks that is several orders of magnitude faster and less memory consuming than established methods for multiclass models, such as the Mean Value Analysis (\MVA) algorithm. The technique is an important generalization of the recently proposed Method of Moments (\MoM) which, differently from \MVA, recursively computes {higher-order} moments of queue-lengths instead of mean values.% using a matrix difference equation involving normalizing constants of the Markov chain underlying the queueing network.

The main contribution of this paper is to prove that the information used in the \MoM recursion can be increased by considering multiple recursive branches that evaluate models with different number of queues. This reformulation allows to formulate a simpler matrix difference equation which leads to large computational savings with respect to the original \MoM recursion. Computational analysis shows several cases where the proposed algorithm is between $1,000$ and $10,000$ times faster and less memory consuming than the original \MoM, thus extending the range of multiclass models where exact solutions are feasible.
%For example, a model with five classes, one thousand users, and ten stations can be solved in x seconds by \MVA, y seconds by \MoM, but requires only z seconds with the \dnc approach.
%~\\
%~\\
%{\em Keywords: Closed Queueing Networks, Multiclass Models, Method of Moments, Exact Analysis, Higher-Order Moments, Recursive Solution, Normalizing Constant}
\end{abstract}

\section{Introduction}
Product-form queueing networks \cite{BasCMP75} are popular stochastic models used in capacity planning of computer architectures and networks with the purpose of evaluating the effect of resource sharing on scalability. In many applications, notably modern multi-tier architectures hosting web sites and intranet applications, workloads are best described as multiclass, that is, requests are assigned to different categories according to the statistical characteristics of their demand at the different servers. Yet, multiclass workloads are extremely challenging to analyze in queueing networks even using state-of-the-art solution techniques such as Mean Value Analysis (MVA) \cite{ReiL80}, the Convolution Algorithm \cite{Buz73,ReiK75b}, RECAL \cite{ConG86}, LBANC \cite{ChaS80}, or more recent methods based on the generating function approach \cite{HarC02,BerM93,ChoLW95}. The main problem is that multiclass models typically involve at least four or five classes, hundreds or thousands of competing requests, and many servers. Yet, established exact solution methods require computational costs which are prohibitive for models of this size, e.g., memory requirements are usually of the order of many terabytes or more. As a result, multiclass networks cannot be usually solved with exact techniques and the focus is on approximation methods\cite{Bar79,Sch79,ChaN82,CreSS02}, which yet cannot return probabilistic measures because they ignore the normalizing constant of the Markov chain underlying the queueing network.

Recently, we have proposed the Method of Moments (MoM) \cite{Cas06b,Cas09}, a new exact technique for multiclass models that recursively computes higher-order moments of queue-length instead of mean values like the MVA approach. The MoM approach is based on normalizing constants, thus it can also compute probabilistic measures that cannot be evaluated by the MVA algorithm. More importantly, the higher-order moments approach is much more scalable that the MVA approach, since the computational costs increase at most log-quadratically with the total population in the network, whereas they grow exponentially with the number of queues or classes in existing methods such as MVA, RECAL, or LBANC. Although much more efficient than MVA, the MoM approach becomes infeasible if the number of queue and classes grows simultaneously \cite{Cas06b}, thus models with many classes and many queues can be hard to analyze even with MoM. In order to address this limitation, we propose in this paper a generalization of MoM. The proposed approach is always more efficient that the original MoM in all cases, yet the largest improvements are achieved on models with several queues and many classes which are infeasible in the original MoM.

Our idea consists in integrating the recursive equation used in the Convolution Algorithm \cite{Buz73,ReiK75b} within the MoM approach, which jointly considers in a linear matrix difference equation the exact recursive formulas for normalizing constants used in RECAL \cite{ConG86} and LBANC \cite{ChaS80}, but not those used in Convolution. By integrating a new formula in the MoM matrix difference equation we obtain a new computational scheme which evaluates higher-order moments of queue-lengths on models with different populations and, as a result of the generalization, also on models with different number of queues. The main advantage of this approach is that the size of the matrix recurrence equation solved at each step of the recursion is much smaller that the one used in the original MoM approach. This is a fundamental improvement since linear system solution required to solve the matrix recurrence grows quadratically or cubically with the coefficient matrix order. In particular, we show that even using a multi-branched recursion on hundreds or thousands of models with different number of queues, the generalized MoM is much more efficient than the original MoM which does not consider models with different number of queues.

The remainder of this paper is organized as follows. After giving background in Section~\ref{SEC:background}, we use in Section~\ref{SEC:motivating} a simple multiclass model to illustrate MoM and the principles of the generalization proposed in this paper. The analysis of the effects of the multi-branched recursion on models with different number of queues is derived in Section~\ref{SEC:dncmom}, where we give in Theorem~\ref{THM:mainM} and Theorem~\ref{THM:main1} the main theoretical results of this paper. Computational complexity of the resulting algorithm is analyzed in Section~\ref{SEC:complexity}. Finally, Section~\ref{SEC:conclusions} gives conclusions and outlines possible extensions of this paper.

\section{Background}
\label{SEC:background}
We consider a closed product-form queueing network with $M$ distinct queues and $R$ service classes. Jobs are routed probabilistically through the queues where they receive service; after completing service, all jobs re-enter the network with a delay of $Z_r$ units of time which depends on the request's service class $r=1,\ldots,R$. The mean service demand, i.e., the mean service time multiplied by the mean number of visits\cite{DenB78}, of class-$r$ jobs at queue $k$ is indicated with $D_{k,r}$. The number of jobs of class $r$ is the integer $N_r$; we define $\vec N=(N_1,N_2,\ldots,N_R)$ as the population vector of the model and $N=N_1+N_2+\ldots+N_R$ is the total number of jobs circulating in the network.

We consider the computation of mean performance indices such as the mean throughput $X_r(\vec N)$ and the mean response time $R_{r}(\vec N)=N_r/X_r(\vec N)$ of class-$r$ jobs; additionally, for each queue $k$ and class $r$, we are interested in computing the utilization $U_{k,r}(\vec N)=D_{k,r}X_r(\vec N)$, the mean queue-length $Q_{k,r}(\vec N)$, and the mean residence times $R_{k,r}(\vec N)=Q_{k,r}(\vec N)/X_r(\vec N)$. These quantities are uniquely determined if one knows how to compute efficiently throughput and mean queue-lengths, which are given by the following ratios \cite{Lam83}:
\begin{align}
\label{EQU:indexes}
X_r(\vec N)=&\frac{G(\vec m, \vec N-\vec 1_r)}{G(\vec m, \vec N)},\\
Q_{k,r}(\vec N)=&\frac{D_{k,r}G(\vec m+\vec 1_k, \vec N-\vec 1_r)}{G(\vec m, \vec N)},
\end{align}
where $G(\vec m, \vec N)$ denotes the normalizing constant of the equilibrium state probabilities of the Markov chain underlying the queueing network\cite{GorN67}, $\vec 1_l$ indicates a vector composed by all zeros except for a one in the $l$th position, and $\vec m\equiv (m_1,m_2,\ldots,m_M)$ is the multiplicity vector such that the multiplicity $m_k$ is the number of queues in the model with identical service demands $D_{k,1},D_{k,2},\ldots,D_{k,R}$. According to these definitions, e.g., $G(\vec m+\vec 1_k, \vec N-\vec 1_r)$ represents the normalizing constant of a model augmented with an additional copy of queue $k$ and with a job of class $r$ removed. Because of the presence of replicated stations, the total number of queues in the model is $M_{tot}=\sum_{k=1}^M m_k$, among which only $M$ have distinct demands.

The advantage of working with normalizing constants instead of mean values is that $G(\vec m, \vec N)$ enables the computation of probabilistic measures that provide fine-grain information about the equilibrium state of the network. For instance, for the case $\vec m=(1,1,\ldots,1)$ where all queues are distinct, the equilibrium state probabilities can be computed as
\begin{equation}
\label{EQU:prob}
\Pr(\vec n_1, \vec n_2, \ldots, \vec n_M)=\frac{\prod_{k=1}^M C(\vec n_k) \prod_{r=1}^R D_{k,r}^{n_{k,r}}}{G(\vec m, \vec N)},
\end{equation}
where $\vec n_k=(n_{k,1},n_{k,2},\ldots,n_{k,R})$, being $n_{k,r}$ the number of class-$r$ jobs in queue $k$ in the considered state, $C(\vec n_k)=(\prod_r n_{k,r}!)/n_k!$, and $n_k=n_{k,1}+n_{k,2}+\ldots+n_{k,R}$. Note that quantities like (\ref{EQU:prob}) cannot be computed neither by the MVA algorithm nor by local iterative approximations \cite{Bar79,Sch79,ChaN82,CreSS02}, thus the normalization constant approach considered in this paper is inherently more general that these methods.

\subsection{Computational Solution}
The analysis of queueing networks can be performed efficiently either by approaches that directly evaluate mean queue-lengths and throughputs in a recursive fashion, such as the Mean Value Analysis (MVA) algorithm \cite{ReiL80}, or by computational methods for the normalizing constants in (\ref{EQU:indexes}), see \cite{BruB80}. The normalizing constant approach is usually slightly more efficient, although it can suffer numerical issues that do not apply to the mean value approach \cite{Lam82}. From a probabilistic perspective, the MVA algorithm and some methods for the normalizing constant, such as the LBANC algorithm \cite{ChaS80}, can be interpreted as a recursive evaluation of {\em mean} queue-lengths\footnote{For normalizing constant methods such as LBANC, the computation focuses on \emph{un-normalized} mean queue-lengths \cite{Lam83}.} over models with different population sizes. Yet, we have recently noted in \cite{Cas06b,Cas09} that recursively evaluating a set of {\em higher-order moments} of queue-lengths can be much more efficient computationally than computing mean values, while still returning the exact solution of the model. The {Method of Moments} (MoM) \cite{Cas06b} is an algorithm that implements this higher-order moment approach and that we generalize for increased efficiency in the next sections; thus we give here a brief overview of the method. Due to limited space and thanks to wide availability of material on the subject, we point to the literature for MVA \cite{ReiL80}, LBANC \cite{ChaS80}, RECAL \cite{ConG86}, and Convolution \cite{Buz73,ReiK75b}; comparative analyses can be found in \cite{Cas06b,BolGMT98}.

\subsubsection{Method of Moments (MoM)}
MoM computes the normalizing constant by simultaneously considering into a linear system of equations the following exact formulas for normalizing constants: the {\em convolution expression} (CE) \cite{ChaS80,Lam83}
\begin{equation}
\label{EQU:ce} G(\vec m+\vec 1_k,\vec N)=G(\vec m, \vec N)+\sum_{r=1}^{R} D_{k,r} G(\vec m+\vec 1_k, \vec N-\vec 1_r)
\end{equation}
for all $1\leq k\leq M$,  and the {\em population constraint} (PC) \cite{ConG86, Cas06b}
\begin{multline}
\label{EQU:pc} N_r G(\vec m, \vec N)=Z_{r} G(\vec m, \vec N-\vec 1_r)\\+\sum_{k=1}^{M} m_kD_{k,r} G(\vec m+\vec 1_k, \vec N-\vec 1_r),
\end{multline}
for all $1\leq r\leq R$, which are also the fundamental recurrence relations employed in the LBANC and RECAL algorithms. These recursions are subject to the following termination conditions: (i) $G(\vec m,\vec N)=0$ if any entry in $\vec N$ or $\vec m$ is negative; (ii) $G(\vec 0,\vec 0)=1$, where $\vec 0=(0,0,\ldots,0)$. In classic algorithms, $G(\vec m,\vec N)$ is obtained by recursively evaluating one between (\ref{EQU:ce}) and (\ref{EQU:pc}) until termination conditions are met. Following this approach, time and space requirements grow roughly as $O(N^R)$ if (\ref{EQU:ce}) is used (e.g., LBANC) and as $O(N^M)$ if (\ref{EQU:pc}) is used (e.g., RECAL). In practice, these costs are often prohibitive since in modeling modern systems it is not difficult to have $N$ of the order of hundreds or thousands and $\min \{M,R\}\geq 5-6$ (see \cite{KouB03} for a recent case study), which make the storage requirement of hundreds of gigabytes regardless of the recursion used.

MoM avoids this memory inefficiency by observing that, if one considers a certain subset of normalizing constants ${\vec V}(\vec N)$, which we call {\em basis}, then this basis can be computed recursively by jointly using (\ref{EQU:ce}) and (\ref{EQU:pc}) to define the matrix difference equation
\begin{equation}
\label{EQU:momrec}
{\bf A}(\vec N) {\vec V}(\vec N) = {\bf B}(\vec N) {\vec V}(\vec N-\vec 1_R),
\end{equation}
where ${\vec V}(\vec 0)$ is known from the termination conditions of (\ref{EQU:ce})-(\ref{EQU:pc}), and the matrices ${\bf A}(\vec N)$ and ${\bf B}(\vec N)$ are square of identical size. The matrices ${\bf A}(\vec N)$ and ${\bf B}(\vec N)$ are defined by the coefficients of the equations (\ref{EQU:ce})-(\ref{EQU:pc}) that relate \emph{all and only} the normalizing constants in ${\vec V}(\vec N)$ with those in ${\vec V}(\vec N-\vec 1_R)$. The basis is:
\begin{multline*}
{\vec V}(\vec N)=\{G(\vec m^\prime,\vec N), G(\vec m^\prime,\vec N-\vec 1_1), \ldots, G(\vec m^\prime,\vec N-\vec 1_{R-1})\\\,|\,\vec m^\prime
= \vec m+(\delta_1,\ldots,\delta_M),~R-1\leq {\textstyle \sum_{k=1}^M} \delta_k\leq R\},
\end{multline*}
which is the set of normalizing constants of models where we have increased the elements of the vector $\vec m$ by $R$ or $R-1$ units in all possible ways and where the models are evaluated over the populations $\vec N$, $\vec N-\vec 1_1$, $\ldots$, $\vec N-\vec 1_{R-1}$. The multiplicity increase operation is equivalent to add new queues to the model and, probabilistically, this can be interpreted as computing binomial moments of queue-lengths in the original queueing network \cite{ConG86,Cas06b,Cas09}; hence one concludes that a recursive computation of ${\vec V}(\vec N)$ is also a recursive evaluation of higher-order moments of queue-length. Indeed, the knowledge of ${\vec V}(\vec N)$ is sufficient to compute all the normalizing constants used in (\ref{EQU:indexes}), see \cite{Cas06b}; thus, computing ${\vec V}(\vec N)$ is equivalent to solve the model.

The interest for (\ref{EQU:momrec}) is that the matrix recursion is linear and does not branch exponentially like (\ref{EQU:ce})-(\ref{EQU:pc}), since we can progressively remove the elements of $\vec N$ without increasing the size of the ${\vec V}(\cdot)$ vectors and until the termination condition ${\vec V}(\vec 0)$ is reached. If the linear system (\ref{EQU:momrec}) is non-singular, one can compute ${\vec V}(\vec N)={\bf A}^{-1}(\vec N){\bf B}(\vec N)$ by an exact solution technique, like exact Gaussian elimination or the Wiedemann algorithm\footnote{See, e.g., the LinBox open source library (\url{http://www.linalg.org}) for a free implementation of the Wiedemann algorithm, exact Gaussian elimination, and other exact methods that can be used to solve the MoM matrix difference equation.} which prevent the critical effects of round-off error accumulation when the recursion is evaluated hundreds or thousands of times and also avoid numerical issues arising in normalizing constant computations  \cite{Cas09}. If the Wiedemann algorithm is used, the computational cost of linear system solution grows quadratically with the basis size and as $O(N^2\log N)$ with respect to the total population, which is typically much less than the $O(N^R)$ and $O(N^M)$ of classic methods. An example illustrating the MoM algorithm is given below, together with intuition on the MoM generalization proposed in this work.

\section{Motivating Example}
\label{SEC:motivating}
We begin by illustrating the structure of (\ref{EQU:momrec}) on a simple queueing network with $M=2$ queues, $R=2$ classes, a population $\vec N=(N_1,N_2)$, and where $\vec m=(1,1)$, i.e., all queues are distinct. To compact notation, let use denote $d_{z,k,s}=(m_k + z)\cdot D_{k,s}$ and $G^{+a,b}_{c,d}=G(\vec m+\vec 1_a+\vec 1_b,\vec N -\vec 1_c -\vec 1_d)$. Then, the linear system (\ref{EQU:momrec}) has the following structure:
\begin{multline}
\label{l:example2A} \arraycolsep 0.1em \scriptsize
\underbrace{\begin{bmatrix}
    1 & -D_{1,1}& \cdot & \cdot  & \cdot & \cdot     & -1 &   \cdot &   \cdot    & \cdot \\
    \cdot & \cdot                 & 1 & -D_{1,1} & \cdot & \cdot & \cdot & \cdot & -1 &  \cdot  \\
        \hline
    \cdot & \cdot                 & 1 & -D_{2,1} & \cdot & \cdot & -1 & \cdot & \cdot &  \cdot  \\
    \cdot & \cdot                & \cdot & \cdot                 & 1 & -D_{2,1} & \cdot & \cdot & -1 &  \cdot  \\
        \hline
    \cdot & -d_{1,1,1} & \cdot          & -d_{0,2,1} & \cdot          & \cdot  & N_1 & -Z_{1} & \cdot &  \cdot  \\
    \cdot & \cdot & \cdot                & -d_{0,1,1} & \cdot          & -d_{1,2,1} & \cdot & \cdot & N_1 &  -Z_{1}  \\
    \hline
    \cdot & \cdot                & \cdot & \cdot                 & \cdot & \cdot & N_2 & \cdot & \cdot &  \cdot  \\
    \cdot & \cdot                & \cdot & \cdot                 & \cdot & \cdot & \cdot & N_2 & \cdot &  \cdot  \\
    \cdot & \cdot                & \cdot & \cdot                 & \cdot & \cdot & \cdot & \cdot & N_2 &  \cdot  \\
    \cdot & \cdot                & \cdot & \cdot                 & \cdot & \cdot & \cdot & \cdot & \cdot &  N_2  \\
\end{bmatrix}}_{{\bf A}(\vec N)}
\underbrace{\begin{bmatrix}
G^{+1,1} \\ G^{+1,1}_1 \\ G^{+1,2} \\ G^{+1,2}_1 \\ G^{+2,2} \\ G^{+2,2}_1 \\ G^{+1} \\ G^{+1}_1 \\ G^{+2} \\ G^{+2}_1
\end{bmatrix}}_{\vec V(\vec N)}
\\= \arraycolsep 0.1em \scriptsize
\underbrace{\begin{bmatrix}
    D_{1,2} & \cdot& \cdot & \cdot  & \cdot & \cdot     & \cdot &   \cdot &   \cdot    & \cdot \\
    \cdot & \cdot& D_{1,2} & \cdot  & \cdot & \cdot     & \cdot &   \cdot &   \cdot    & \cdot \\
        \hline
    \cdot & \cdot& D_{2,2} & \cdot  & \cdot & \cdot     & \cdot &   \cdot &   \cdot    & \cdot \\
    \cdot & \cdot& \cdot & \cdot  & D_{2,2} & \cdot     & \cdot &   \cdot &   \cdot    & \cdot \\
        \hline
    \cdot & \cdot& \cdot & \cdot  & \cdot & \cdot     & \cdot &   \cdot &   \cdot    & \cdot \\
    \cdot & \cdot& \cdot & \cdot  & \cdot & \cdot     & \cdot &   \cdot &   \cdot    & \cdot \\
        \hline
    d_{1,1,2} & \cdot& d_{0,2,2} & \cdot  & \cdot & \cdot     & Z_{2} &   \cdot &   \cdot    & \cdot \\
    \cdot & d_{1,1,2} & \cdot& d_{0,2,2} & \cdot  & \cdot & \cdot     & Z_{2} &   \cdot &   \cdot  \\
    \cdot & \cdot & d_{0,1,2} & \cdot  & d_{1,2,2} & \cdot     & \cdot & \cdot&  Z_{2} &   \cdot    \\
    \cdot & \cdot & \cdot & d_{0,1,2}  & \cdot & d_{1,2,2}     & \cdot & \cdot&  \cdot &  Z_{2}    \\
\end{bmatrix}}_{{\bf B}(\vec N)}
\underbrace{\begin{bmatrix}
G^{+1,1}_2 \\ G^{+1,1}_{1,2} \\ G^{+1,2}_{2} \\ G^{+1,2}_{1,2} \\ G^{+2,2}_{2} \\ G^{+2,2}_{1,2} \\ G^{+1}_{2} \\ G^{+1}_{1,2} \\ G^{+2}_{2} \\ G^{+2}_{1,2}
\end{bmatrix}}_{\vec V(\vec N-\vec 1_R)}
\nonumber
\end{multline}
where $\cdot$ indicates a zero element, and the four blocks of the coefficient matrices represent from top: the CE (\ref{EQU:ce}) for $k=1$,  the CE for $k=2$, the PC (\ref{EQU:pc}) for $r=1$, and the PC for $r=2$. The basis of normalizing constants is depicted in Figure~\ref{FIG:basis}. We remark that, for each element in the figure, the basis includes both normalizing constants for the populations $\vec N$ and $\vec N-\vec 1_1$, hence the total number of elements is $card({\vec V}(\vec N))=10$.

\begin{figure}
\center
  \includegraphics[width=7cm]{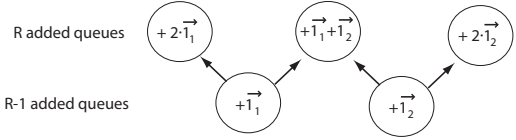}\\
  \caption{\footnotesize \em Basis of normalizing constants ${\vec V}(\vec N)$ for a model with $M=2$ queues and $R=2$ classes. Each circle represents a group of $R$ normalizing constants of models with populations $\vec N$ and $\vec N-1_{1}$. Labels indicate the increase of the multiplicity vector $\vec m$ relatively to that subset of normalizing constants.}
  \label{FIG:basis}
\end{figure}

\noindent The fundamental observations presented in this paper to improve MoM and, specifically, to considerably reduce the cost of computing ${\vec V}(\vec N)$, are as follows:
\begin{enumerate}
 \item we first note that it is possible to add independent equations to the above linear system by taking in consideration a generalization of the convolution expression (\ref{EQU:ce}) explained later in the paper; this generalization provides independent information and makes the linear system over-determined.
 \item we show that, if the linear system is over-determined, then the basis ${\vec V}(\vec N)$ can be defined smaller, while still preserving the capability of MoM of solving exactly queueing networks. The basis size reduction leads to remarkable computational savings compared to the original MoM approach.
 \item however, as we explain in Section~\ref{SEC:dncmom}, for models of arbitrary size the additional independent information comes at the price of additional recursions over models with different number of queues. We investigate in the rest of the paper if accepting these additional recursions is convenient with respect to the computational savings implied by the basis size reduction.
\end{enumerate}
The previous observations are further illustrated in the next subsection.

\subsection{Improved Computation of the Basis of Normalizing Constants}

We begin by observing that (\ref{EQU:ce}) can be seen as a specialization of the recursive equation used by the Convolution Algorithm\cite{Buz73,ReiK75b}, which we call the {\em generalized convolution expression} (GCE)
\begin{equation}
\label{EQU:gce} G(\vec m,\vec N)=G(\vec m-\vec 1_k, \vec N)+\sum_{r=1}^{R} D_{k,r} G(\vec m, \vec N-\vec 1_r),
\end{equation}
for all $1\leq k\leq M$. Here queues are removed through the parameter $\vec m-\vec 1_k$, instead of being added as in (\ref{EQU:ce}). This implies that a recursion involving (\ref{EQU:gce}) may also evaluate models which contain less queues than in the original queueing network, while (\ref{EQU:ce}) operates on networks with multiplicity $\vec m^\prime \geq \vec m$ only. However, by instantiating  (\ref{EQU:gce}) on a model with multiplicity $\vec m+\vec 1_k$ instead of  $\vec m$, it is found that (\ref{EQU:gce}) becomes identical to (\ref{EQU:ce}), thus (\ref{EQU:ce}) specifies a subset of (\ref{EQU:gce}). Whenever (\ref{EQU:gce}) is instantiated on models with less queues than in the original network, the information provided by (\ref{EQU:gce}) is independent with respect to the one provided by (\ref{EQU:ce}), because the two equations are defined over models with different network structure. For example, equation (\ref{EQU:gce}) may be added to the simple queueing network considered before if instantiated as
\begin{multline}
\label{EQU:addsce}
G(\vec m+2\cdot \vec 1_1,\vec N)=G(\vec m+2\cdot \vec 1_1-\vec 1_2, \vec N)\\+\sum_{r=1}^{R} D_{2,r} G(\vec m+2\cdot \vec 1_1, \vec N-\vec 1_r).
\end{multline}
In this case, the normalizing constant $G(\vec m+2\cdot \vec 1_1-\vec 1_2, \vec N)$ lies outside the basis ${\vec V}(\vec N)$, thus equation (\ref{EQU:addsce}) does not reduce to a CE and provides independent information. Note also that $G(\vec m+2\cdot \vec 1_1-\vec 1_2, \vec N)$ is the normalizing constant of a model where queue $2$ has been completely removed since we have assumed $\vec m=(1,1)$, thus it can be computed easily with closed-form formulas for the balanced network case \cite{MitM85} and therefore the addition of (\ref{EQU:addsce}) does not increase the number of unknowns in the linear system.

The main idea investigated in this paper is that this independent information can be exploited effectively to reduce the size of the basis ${\vec V}(\vec N)$. In fact, consider a new basis ${\vec V}_{new}(\vec N)$ composed by normalizing constants with $R-2\leq \sum_k m_k\leq R-1$ additional queues instead of the $R-1\leq \sum_k m_k\leq R$ as in the original definition of ${\vec V}(\vec N)$. Then, using (\ref{EQU:ce}), (\ref{EQU:pc}), and (\ref{EQU:addsce}), we can define a linear system with square matrix of coefficients
\begin{multline}
\label{l:example2A} \arraycolsep 0.1em \scriptsize
\underbrace{\begin{bmatrix}
    1 & -D_{1,1}& \cdot & \cdot  & -1  & \cdot \\
    \hline
    \cdot & \cdot  & 1 & -D_{2,1}& -1  & \cdot \\
    \hline
    1 & -D_{2,1}& \cdot & \cdot  & \cdot  & \cdot \\
    \hline
    \cdot & -d_{0,1,1} & \cdot & -d_{0,2,1} & N_1  & -Z_1 \\
    \hline
    \cdot & \cdot & \cdot & \cdot  & N_2  & \cdot \\
    \cdot & \cdot & \cdot & \cdot  & \cdot  & N_2 \\
\end{bmatrix}}_{{\bf A}_{new}(\vec N)}
\underbrace{\begin{bmatrix}
G^{+1} \\ G^{+1}_1 \\ G^{+2} \\ G^{+2}_1 \\ G \\ G_1
\end{bmatrix}}_{\vec V_{new}(\vec N)}
\\=  \arraycolsep 0.1em \scriptsize
\underbrace{\begin{bmatrix}
    \cdot \\
    \hline
        \cdot \\
    \hline
    G^{+1,+1,-2} \\
    \hline
    \cdot \\
    \hline
    \cdot \\
    \cdot \\
\end{bmatrix}}_{\vec V^{-k}_{new}(\vec N)}
+
\underbrace{\begin{bmatrix}
    D_{1,2} & \cdot & \cdot & \cdot & \cdot & \cdot \\
    \hline
    \cdot   & \cdot & D_{2,2} & \cdot & \cdot & \cdot \\
    \hline
    D_{2,2} &\cdot   & \cdot & \cdot & \cdot & \cdot \\
    \hline
    \cdot &\cdot   & \cdot & \cdot & \cdot & \cdot \\
    \hline
    d_{0,1,2} &\cdot   & d_{0,2,2} & \cdot & Z_2 & \cdot \\
    \cdot & d_{0,1,2} &\cdot   & d_{0,2,2} & \cdot & Z_2 \\
\end{bmatrix}}_{{\bf B}_{new}(\vec N)}
\underbrace{
\begin{bmatrix}
G^{+1}_2 \\ G^{+1}_{1,2} \\ G^{+2}_2 \\ G^{+2}_{1,2} \\ G_2 \\ G_{1,2}
\end{bmatrix}}_{\vec V_{new}(\vec N-\vec 1_R)}
\nonumber
\end{multline}
where the new vector $\vec V^{-k}_{new}(\vec N)$ includes the normalizing constant $G^{+1,+1,-2}\equiv G(\vec m+2\cdot \vec 1_1-\vec 1_2, \vec N)$ used in (\ref{EQU:addsce}), and the blocks of the coefficient matrix are from the top: the CE for $k=1$, the CE for $k=2$, the GCE (\ref{EQU:addsce}), the PC for $r=1$, and the PC for $r=2$. The new linear system may be written compactly as
\begin{equation}
\label{EQU:momrecnew}
{\bf A}_{new}(\vec N)\vec V_{new}(\vec N)=V^{-k}_{new}(\vec N)+{\bf B}_{new}(\vec N)\vec V_{new}(\vec N-\vec 1_R)
\end{equation}
with square coefficient matrix, thus if ${\bf A}_{new}^{-1}(\vec N)$ exists \emph{the solution of the linear system (\ref{EQU:momrecnew}) provides an alternative way to recursively compute normalizing constants that is cheaper than the original linear system (\ref{EQU:momrec}), since (\ref{EQU:momrecnew}) halves the order of the coefficient matrix with respect to (\ref{EQU:momrec})}. We stress that without (\ref{EQU:addsce}) the new system (\ref{EQU:momrecnew}) would be under-determined, thus resorting to the GCE equations is critical for this new approach.

It is also important to remark that, for queueing networks larger than the one considered in this experiment, the normalizing constants in $V^{-k}_{new}(\vec N)$ may not be available from closed-form expressions. In this case, the computation of $V^{-k}_{new}(\vec N)$ requires additional recursions over models with different number of queues; we show in the next section that, if multiple equations (\ref{EQU:gce}) are used simultaneously, this yields a multi-branched recursive structure for the MoM algorithm, where one needs to evaluate recursively also models with less queues that are not considered in the original MoM recursion.

\section{The Multi-Branched Method of Moments}
\label{SEC:dncmom}
The integration of the GCE (\ref{EQU:gce}) into the MoM linear system can be done in different ways depending on the number of equations (\ref{EQU:gce}) simultaneously instantiated into the matrix difference equation. As observed earlier, integrating GCEs into MoM allows to reduce the basis size; this reduction can be specified by a decrease in the number of queues added to the multiplicity vectors in the basis, which is equivalent to considering queue-length moments of smaller order. Specifically, if the new basis $V_{new}(\vec N)$ includes models with only $l-1$ and $l$ added queues, one can integrate a single or multiple GCEs \emph{for each model with $l$ additional queues only}\footnote{The GCE is not needed for models with $l-1$ additional queues, since their normalizing constants are all immediately computed from the basis for $\vec V(\vec N-1_R)$ using the PC of class $R$.}. A comparison of the recursion trees arising from the two alternatives (single or multiple GCEs) is given in Figure \ref{FIG:singlevsmultiple}.

\begin{figure}
\center
  \subfigure[Single GCE]{\includegraphics[width=2.3cm]{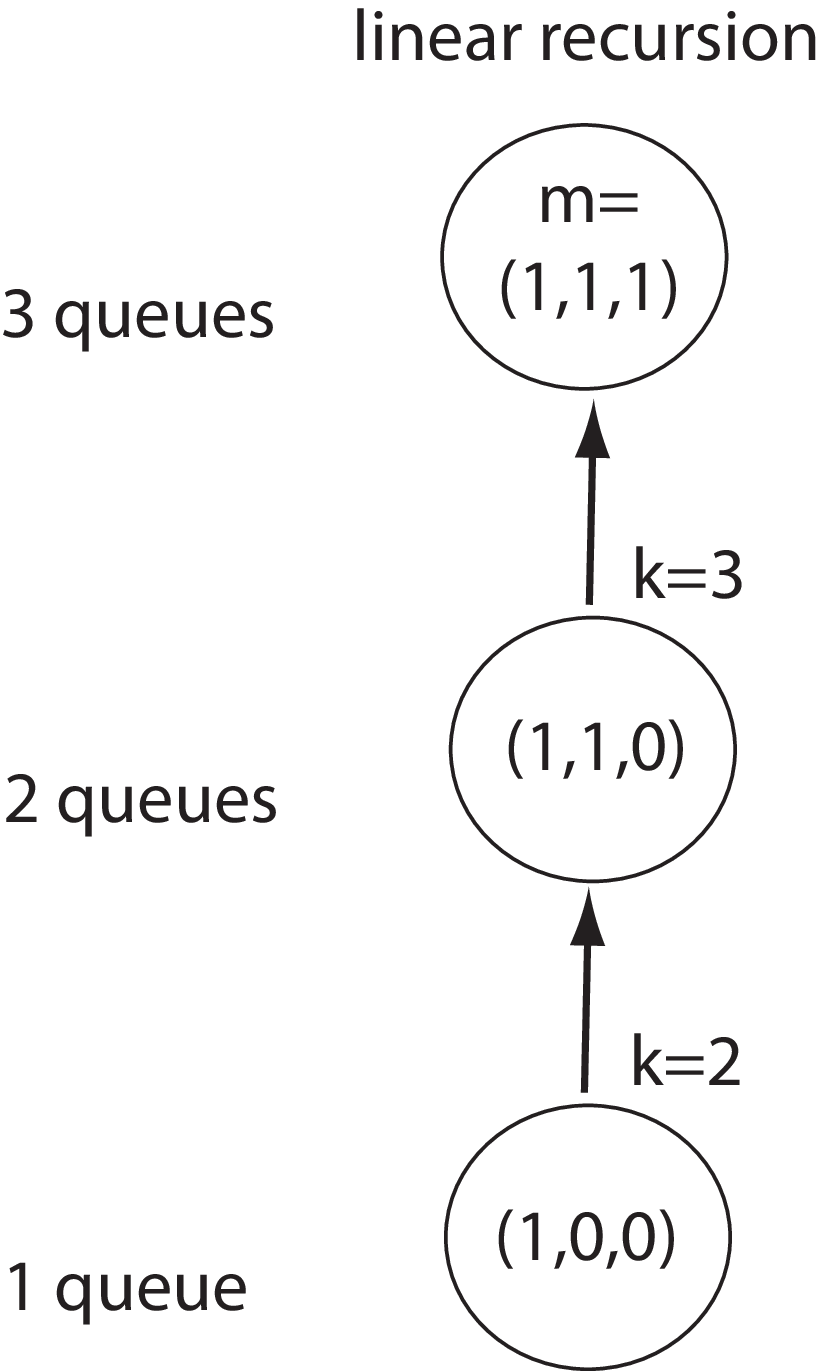}}\quad\quad\quad
  \subfigure[Multiple GCEs]{\includegraphics[width=4cm]{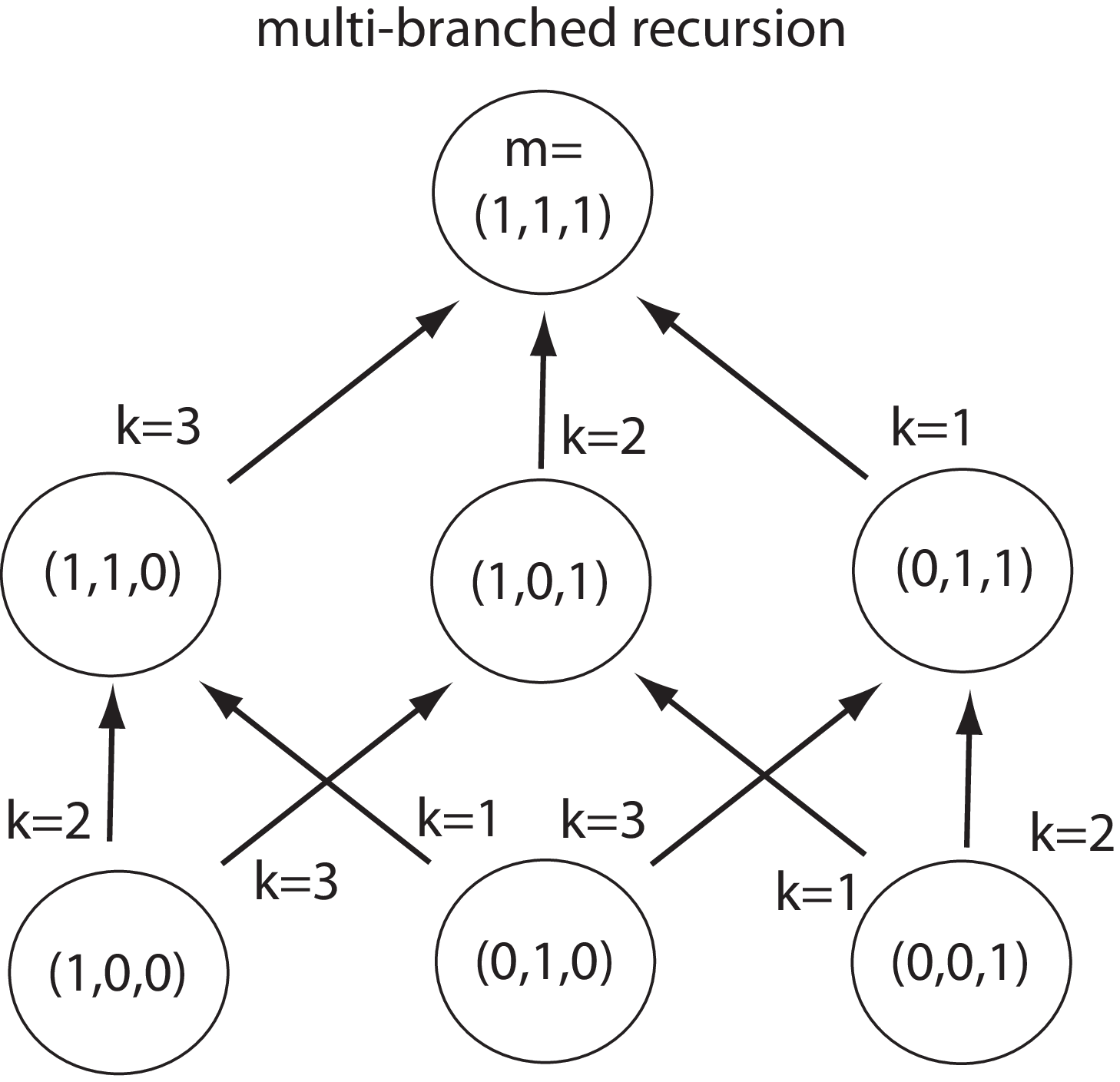}}\\
  \caption{\footnotesize \em Structure of the MoM recursion after addition of a single or multiple GCEs (\ref{EQU:gce}) on a model with $M=3$ queues. The label $k=1$, e.g., indicates a GCE instantiated for $k=1$ on all models with $l$ added queues in the redefined basis. Labels within a circle indicate the multiplicity vector $\vec m$ on which the basis is defined, e.g., $(1,0,1)$ is the model obtained from the original queueing network by removing queue $2$.}
  \label{FIG:singlevsmultiple}
\end{figure}

Using a {single} GCE implies an additional $M$ recursions which first remove from the model queue $M$, followed by queue $M-1$, and so forth up to a trivial model with a single queue. Instead, using all possible GCEs implies that the additional recursions first consider {\em all} possible ${M\choose M-1}$ models with $M-1$ queues, followed by {\em all} possible ${M\choose M-2}$ models with $M-2$ queues, and so on up to models with a single queue. The latter approach appears to be the most expensive, at least if one ignores the basis size reduction, because it has a number of new recursions that grows combinatorially instead of linearly. Yet, while limiting to a single GCE seems a natural choice to control the number of new recursions, we have noted that in practice the additional information of the multiple GCEs implies a much larger reduction of the basis than in the case of a single GCE. This in turn provides computational savings often greater than the additional overheads imposed by the extra recursions. Thus, in this section we investigate the trade-off imposed by different types of integrations of GCEs and consider the general case of simultaneously considering up to $B$, $1\leq B\leq M$, GCEs in the linear system. We also provide a complexity analysis to evaluate the best choice of this {\em branching factor} $B$ as a function of the other model parameters.

\subsection{Basis Reduction}
We now investigate the reduction of the basis cardinality as a function of the number of GCE equations added to the MoM matrix difference equation. Indeed, the most interesting cases are 1) when (\ref{EQU:gce}) is added for a single value of $k$ or 2) when all possible equations in (\ref{EQU:gce}) are added; in fact, intermediate cases imply a combinatorial branching of the recursion and thus grow in computational complexity similarly to the second case. Let us define a {\em basis of level $l$}, $l\geq 1$, as the set
\begin{multline}
{\vec V}_l(\vec N)=\{G(\vec m^\prime,\vec N), G(\vec m^\prime,\vec N-\vec 1_1), \ldots, G(\vec m^\prime,\vec N-\vec 1_{R-1})
\\\,|\,\vec m^\prime= \vec m+(\delta_1,\ldots,\delta_M),~l-1\leq {\textstyle \sum_{k=1}^M} \delta_k\leq l\},
\end{multline}
which is the set of normalizing constants with $l-1$ or $l$ additional replicated queues. According to this definition, in MoM it is always ${\vec V}(\vec N)\equiv {\vec V}_R(\vec N)$, while the basis in the example of the last section after the addition of the GCE is $\vec V_{new}(\vec N)\equiv {\vec V}_{R-1}(\vec N)$. A basis of level $l$ has cardinality $card({\vec V}_l(\vec N))={M+l-1\choose l}R$, thus \emph{a decrease, thanks to the GCEs, of $l$ even by a few units implies a quick combinatorial reduction of the number of elements in the basis}. The next theorems are the fundamental result of this paper and exactly quantify the amount of this reduction.

%{\bf G:Explain better that all GCEs means all wrt the models with $l$ replicas }
\begin{theorem}
\label{THM:mainM}
The inclusion in the MoM matrix difference equation (\ref{EQU:momrec}) of the GCEs (\ref{EQU:gce}) for $k=1,\ldots,M$ on all models having $l$ additional queues in the basis ${\vec V}_l(\vec N)$ allows to define a linear system of the type
\begin{equation}
\label{EQU:dcmomrec}
{\bf A}_{l}(\vec N)\vec V_{l}(\vec N)=\vec V^{-k}_{l}(\vec N)+{\bf B}_{l}(\vec N)\vec V_{l}(\vec N-\vec 1_R),
\end{equation}
which has more equations than unknowns if $l\geq \max\{1,R-M\}$. Therefore, the basis has minimum cardinality for $l=\max\{1,R-M\}$.
\end{theorem}
\begin{proof}
A basis of level $l$ has ${M+l-1\choose l}R$ normalizing constants, while the total number of CEs and PCs is ${M+l-2\choose l-1}(M+R-1)$ since there exist $M$ CEs and $R-1$ PCs for each of the ${M+l-2\choose l-1}$ normalizing constant with $l-1$ additional queues and all other possible CEs and PCs require constants outside the basis. Thus we have that, in absence of GCEs, there are more equations than unknowns in the matrix difference equation if
$
{M+l-2\choose l-1}(M+R-1) \geq {M+l-1\choose l}R,
$
which is true for all $l\geq R$. In particular, $l=R$ gives the minimum cardinality of the basis and for this reason it is the choice done by MoM for its basis ${\vec V}(\vec N)\equiv {\vec V}_R(\vec N)$ which ignores the GCEs. We now add to the previous condition the number of additional GCEs which do not specialize into CEs and that we can formulate for models with $l$ additional queues, which is
$
\sum_{h=1}^{\min\{M,l\}} {M\choose h}{l-1\choose l-h}(M-h).
$
This can be explained as follows. Consider a model with normalizing constant in ${\vec V}_l(\vec N)$ and where we have added $l$ queues. Denote by $h$ the number of distinct queues among the $l$ queues we have added. It is possible to see that removing any of these $h$ queues using a GCE involves only normalizing constants with $l-1$ added queues that are already in the bases ${\vec V}_l(\vec N)$ and ${\vec V}_l(\vec N-\vec 1_R)$, thus these specific GCE equations are identical to the CEs and do not provide independent information. Therefore, for a model with $h$ distinct additional queues, only $M-h$ GCEs are different from the existing CEs. Note that there are ${M\choose h}$ ways of choosing the $h$ distinct queues and ${h+(l-h)-1\choose l-h}={l-1\choose l-h}$ ways of adding $l$ queues to the model chosen among these $h$ distinct ones under the constraint that each of the $h$ queues is chosen at least once. Combining these expressions gives the number of GCEs that are not CEs, which simplifies to
$
\sum_{h=1}^{\min\{M,l\}} {M\choose h}{l-1\choose l-h}(M-h)
\\={M+l-2\choose l-1}(M+R-1)={M+l-2\choose l}M,
$
where the first passage follows by Vandermonde convolution \cite{Coh78}.

Adding the number of GCEs that are not CEs, we evaluate the following condition for (\ref{EQU:dcmomrec}) to have more equations than unknowns
$
{M+l-2\choose l}M+{M+l-2\choose l-1}(M+R-1) \\\geq {M+l-1\choose l}R.
$
Suppose first $R>M+1$ and thus $l=\max\{1,R-M\}=R-M$, then we consider the condition
$
{R-2\choose R-M}M+{R-2\choose R-M-1}(M+R-1)\geq{R-1\choose R-M}R
$
which using the property of binomial coefficients ${n\choose k}={n-1\choose k}+{n-1\choose k-1}$ on the right hand side gives
${R-2\choose R-M}M+{R-2\choose R-M-1}(M+R-1)\geq{R-2\choose R-M}R + {R-2\choose R-M-1}R$
that simplifies to
${R-2\choose R-M-1}(M-1)\geq{R-2\choose R-M}(R-M)$
which is actually an equality because, after expanding the binomial coefficients, both sides are found identical. Hence, since $l=R-M$ always returns an equality between number of equations and number of unknowns, it is easy to verify that $l<R-M$ would always give an under-determined system and thus $l=R-M$ gives the minimum allowable basis size for the case $R>M+1$.

Consider now the other case $R\leq M+1$ where $l$ takes the minimum possible value $l=\max\{1,R-M\}=1$, we have then
$
{M\choose 1}(M-1)+{M-1\choose 0}(M+R-1)\geq{M\choose 1}R
$
which is equivalent to
$M(M-1)+(M+R-1)\geq MR$ and assuming the worst case $R=M+1$ we get
$M(M-1)+2M\geq M(M+1)$
which is always true because the two sides simplify to the same identical value. This means that the linear system is always square if we use the minimum value $l=1$ when $R\leq M+1$.

We can summarize the above findings saying that $l=\max\{1,R-M\}$ always implies a ${\bf A}_{l}(\vec N)$ matrix that is square or over-determined and that smaller values of $l$ instead result in under-determined systems for certain values of $M$ and $R$. This concludes the proof of the theorem.
\end{proof}

\begin{theorem}
\label{THM:main1}
The inclusion of a single GCE (\ref{EQU:gce}) for given $k$ in the MoM matrix difference equation (\ref{EQU:momrec}) allows to define a linear system similar to (\ref{EQU:dcmomrec}), but which has more equations than variables if $l\geq \max\{1,R-1\}$. In particular, the basis has minimum cardinality for $l=\max\{1,R-1\}$.
\end{theorem}
\begin{proof}
The proof differs from that of Theorem~\ref{THM:mainM} for the number of GCE equations that are not CEs. Suppose that GCEs for given $k$ are used, and assume without loss of generality that the GCE of station $k=M$ is the one included in the matrix difference equation. Then the number of GCEs that are not CEs is
\begin{equation*}
\sum_{h=1}^{\min\{M,l\}} {M\choose h}{l-1\choose l-h}-\sum_{h=1}^{\min\{M-1,l-1\}} {M\choose h}{l-1\choose l-h},
\end{equation*}
where the left term follows similarly to the number of GCEs in Theorem~\ref{THM:mainM}, but for the case where one GCE is added, instead of $M$, to the models in $\vec V_l(\vec N)$ with $l$ additional queues. The right term counts instead the number of times this GCE is identical to an existing CE. The above expression becomes simpler thanks to Vandermonde convolution \cite{Coh78} and gives that we have more equations than unknowns if
$
{M+l-2\choose l}+{M+l-2\choose l-1}(M+R-1) \geq {M+l-1\choose l}R.
$
Now using ${n\choose k}={n-1\choose k}+{n-1\choose k-1}$ on the right hand side we get
$
{M+l-2\choose l}+{M+l-2\choose l-1}(M+R-1) \geq {M+l-2\choose l}R+{M+l-2\choose l-1}R,
$
which is equivalent to
$
{M+l-2\choose l-1}(M-1) \geq {M+l-2\choose l}(R-1).
$
and expanding the binomial coefficients it is found that the two sides are identical if $l=\max\{1, R-1\}$ which completes the proof.
\end{proof}
The results in Theorem~\ref{THM:mainM} and Theorem~\ref{THM:main1} show that: (1) if all GCEs are added to the MoM matrix difference equation, then the basis level can be decreased by up to $M$ units; (2) if a single GCE is used, the basis level can instead be decreased by a single unit. Following the same line of the proofs of Theorem \ref{THM:mainM} and Theorem \ref{THM:main1} it is then straightforward to show the following corollary.

\begin{corollary}
\label{COR:impact}
If $B$ GCEs, $1\leq B\leq M$, are added to the matrix difference equation, then the basis can be decreased by up to $B$ levels and the minimal basis size is obtained with the basis level $l=\min\{1,R-B\}$.
\end{corollary}
\noindent The next section investigates the computational implications of the last result.

\section{Computational Complexity}
\label{SEC:complexity}
Corollary~\ref{COR:impact} enables the evaluation of the optimal choice of the \emph{branching factor} $B$ as a function of the model size. In practice, we are interested to understand when the additional recursions implied by a branching factor $B$ give an overhead that is less that the savings implied by the reduction of the basis level from $l=R$ of the original MoM to $l=\min\{1,R-B\}$ of MoM with GCEs.

We first observe that if $B$ GCEs are used in the MoM linear system, then the basis $\vec V^{-k}_{l}(\vec N)$ in (\ref{EQU:dcmomrec}) is computed recursively from $B$ bases of models with a queue less. These models have $M-1$ queues, thus the branching factor in this case is upper bounded by $B\leq M-1$. That is, the maximum number of GCEs added to the linear system changes according to the distance $d$, $d=0,\ldots,{M-1}$, in the recursion tree from the root (i.e., the original model). For $d=0,\ldots,B-1$, only up to $d$ GCEs can be added to the linear system, while for distances $d=B,\ldots,M$ we can always add $B$ GCEs. This can be seen immediately from Figure~\ref{FIG:singlevsmultiple}, where the number of GCE equations instantiated for a model with three queues are three ($k=1$, $k=2$, and $k=3$), two for a model with two queues, and they decrease progressively during the recursion.

Starting from the previous consideration, we analyze below the computational complexity of MoM with GCEs for the two limit cases $B=1$ and $B=M$, and provide discussion about the intermediate cases $1<B<M$ at the end of this subsection.

\begin{figure*}[t]
 \centering
    \includegraphics[width=\textwidth]{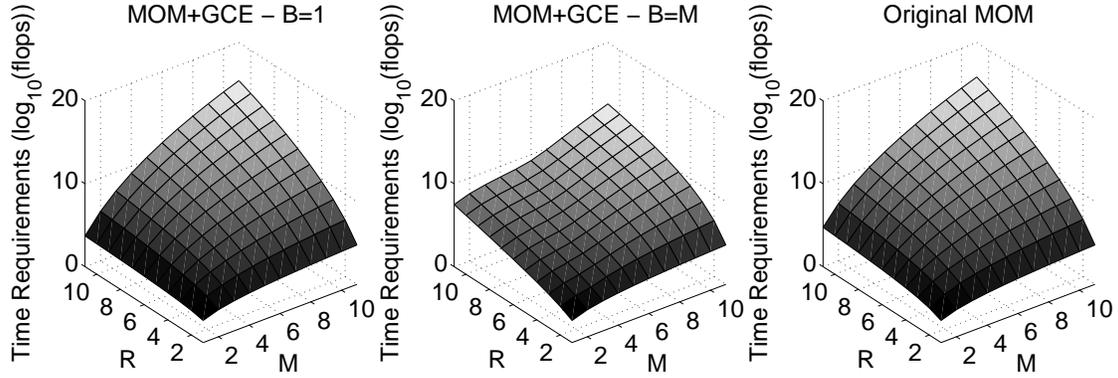}\\
  \caption{\footnotesize Time requirements of MoM and the divide-and-conquer MoM for different number of queues $M$ and number of service classes $R$. All queues are assumed distinct. The results indicate that assuming a branching level $B=M$ is far superior to $B=1$, unless a small number of queues is considered in the model ($M\leq 4$). The total population in the network is set to $N=100$.}
  \label{FIG:timerequirements}
\end{figure*}

{\em Time Requirements}. If $B=1$, then $V^{-k}_{l}(\vec N)$ is computed by $M$ recursions. During the $d$th recursive step $d=0,\ldots,M-1$, the model has $M-d$ queues; the basis is always of level $l=R-1$ for all steps. Assuming quadratic costs in the solution of the linear system, e.g., using a method like the Wiedemann algorithm, we have that the time for computing $\vec V(\vec N)$ from $\vec V(\vec N-\vec 1_R)$ grows as
\begin{equation}
\label{EQU:costB1}
\sum_{d=0}^{M-1} \left({M-d+R-2\choose R-1}R\right)^2 S^d_{exact}, % \leq M\left({M+R-2\choose R-1}R\right)^2
\end{equation}
where the term between parenthesis is the coefficient matrix order in (\ref{EQU:dcmomrec}) and
$$S^d_{exact}\approx (N\log(M-d+N)){M-d+R-2\choose R-1}R$$
is the overhead of exact algebra for a model with $M-d$ queues and assuming that the linear system solver uses multiprecision arithmetic \cite{Cas06b}. In the expression (\ref{EQU:costB1}) we have ignored the exact number of iterations of the solution algorithm and thus the expression may be regarded as a cost per iteration of the linear system solver.
%where the upper bound is tight if the number of classes $R$ is large.

In the case where we use all possible GCEs, it is $B=M$ at the first recursive step, then $B=M-1$ at the second step, and $B=M-d$ at the $d$th recursive step\footnote{This observation holds true under the assumption that the model is composed initially by queues that have all multiplicity $m_k=1$, i.e., which are all distinct. The case of models with replicated queues has more favorable computational costs if the total number of queues (including the non-replicated ones) is the same, thus our analysis is a worst-case scenario when $M$ is interpreted as the total number of queues instead of the number of distinct ones.}. In addition, the level used at the $d$th step of the recursion is $l\equiv l(d)=\max\{1,R-M+d\}$, which is thus a function of the distance $d$ from the root of the recursion tree. Following these observations, the time requirements grow as
\begin{equation*}
\sum_{d=0}^{M-1} {M\choose M-d}\left({M-d+l(d)-1 \choose l(d)} R\right)^2 S^d_{exact}
\end{equation*}
where $l(d)=\max\{1,R-M+d\}$ and the term ${M\choose M-d}$ accounts for the combinatorial branching of the recursion and is the number of all possible queueing network models with $M-d$ distinct queues chosen among the initial $M$. For example, when $M<R$
\begin{equation*}
\sum_{d=0}^{M-1} {M\choose M-d}\left({R-1 \choose R-M+d} R\right)^2 S^d_{exact}
\end{equation*}
which is significantly smaller than (\ref{EQU:costB1}) since the binomial coefficient does not longer depend on the sum of $M$ and $R$.

Similarly to the case $B=1$, the time requirements expression is a cost per solver iteration and the term raised to square is the linear system order. Compared to the above expressions, the original MoM algorithm has a time requirement per iteration of
\begin{equation}
\label{EQU:costMOM}
\left({M+R-1\choose R}R\right)^2 S^d_{exact}.
\end{equation}

Figure \ref{FIG:timerequirements} quantifies the savings per solver iteration of the new algorithm for $B=1$ and $B=M$ compared to the costs of the original MoM. Since the costs are dependent on $M$, $R$, and the population size $N$, we simplify the evaluation and consider the variation of $M$ and $R$ under a quite large $N=100$. The cost surfaces indicate that the algorithm with $B=M$ is typically the most efficient except for very low values of $M$ where it is much more expensive than $B=1$ and the original MoM, although the cost per iteration remains quite small. Overall, the savings of the $B=1$ case are quite limited compared to the original MoM, while massive cost reduction is achieved with the multi-branched case $B=M$. This is quite counter-intuitive, since one would at first expect that the wide recursion tree in Figure \ref{FIG:singlevsmultiple}(b) is a major source of computational cost compared to the linear recursive structure in Figure \ref{FIG:singlevsmultiple}(a). Yet, Figure \ref{FIG:timerequirements} indicates that, for multiclass models that can be solved in acceptable times with commonly-available hardware, the cost of the combinatorial branching in Figure \ref{FIG:singlevsmultiple}(b) is not yet a performance bottleneck and it is justified by the massive computational saving of the basis size reduction.

\begin{figure*}[t]
 \centering
    \includegraphics[width=\textwidth]{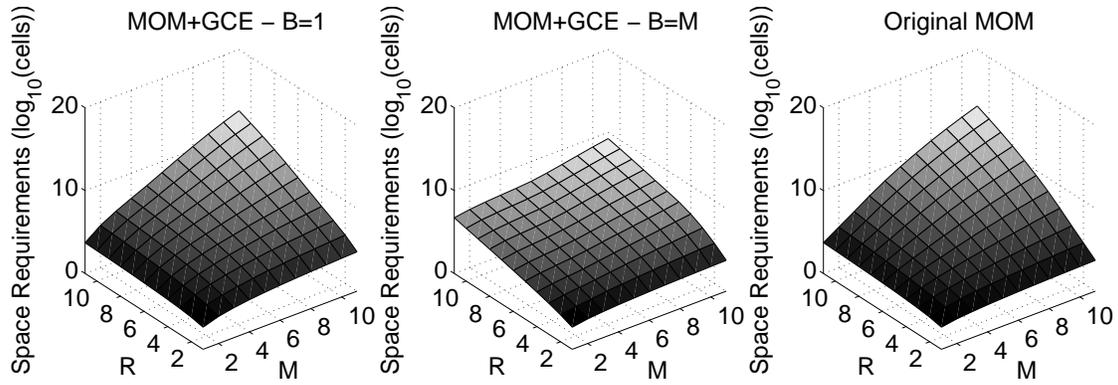}\\
  \caption{\footnotesize Space requirements of MoM and the divide-and-conquer MoM for different number of queues $M$ and number of service classes $R$. The interpretation of the results is qualitatively similar to the one for the time requirements in Figure~\ref{FIG:timerequirements}, with the best branching level being $B=M$ unless the number of queues $M$ is small. The total population in the network is set to $N=100$.}
  \label{FIG:spacerequirements}
\end{figure*}

{\em Space Requirements}. The space requirement of the case $B=1$ is upper bounded by the cost of storing the linear system (\ref{EQU:dcmomrec}) in memory when it is largest, i.e., for the original model with $M$ queues. This is approximately given by
\begin{equation}
2\left({M+R-2\choose R-1}R\right)^2+3\left({M+R-2\choose R-1}R\right)S^{R-1}_{nc}.
\end{equation}

The evaluation of memory requirements in the case $B=M$ is similar, but requires to take into account the width of the multi-branched recursion tree, since all basis vectors for models with $k-1$ queues should be available before the evaluation of models with $k$ queues. Thus, the memory occupation is
\begin{multline}
\max_{d=1,\ldots,M} {M\choose M-d} \Biggl(2\left({M-d+R-2\choose R-1}R\right)^2\\+2\left({M-d+R-2\choose R-1}R\right) S_{nc}\Biggr),
\end{multline}
where the first term is the cost of storing ${\bf A}(\vec N)$ and ${\bf B}(\vec N)$ for the currently evaluated linear system, while the second term accounts for the basis for populations $\vec N$ and $\vec N-\vec 1_R$ of all models at distance $d$ from the root of the recursion.

Finally, the computational costs of the original MoM are given by \cite{Cas06b}
\begin{equation}
2\left({M+R-1\choose R}R\right)^2+3\left({M+R-2\choose R-1}R\right)S^{R}_{nc},
\end{equation}
which is quite similar to the cost of the case $B=1$.

The comparison of the space requirements of the three different methods is shown in Figure~\ref{FIG:spacerequirements} for different values of $M$ and $R$; we set again the total population to $N=100$. Results are qualitatively similar to the time requirement case: the GCE equations provide the largest savings in space requirements compared to the original MoM only if $B=M$. The case $B=1$ is $1-2$ orders of magnitude faster then the original MoM for models with few queues ($M\leq 4$), while as $M$ increases the algorithm with $B=M$ scales much better. In particular, for the most challenging model with $M=11$ and $R=11$, the computational saving of the modified algorithm with $B=M$ is about four orders of magnitude over the original MoM, thus making the case that the inclusion of the GCE equations is highly-valuable also for the space requirements.

{\em Intermediate cases $1<B<M$}. Following the result in Corollary \ref{COR:impact} it is immediately found that the size of the basis for intermediate choices of the branching level $B$ is always bounded by the choices $B=1$ and $B=M$ and computational requirements are typically within those of these limit cases. For example, assume that $B$ queues are chosen for removal and the multi-branched recursion is operated only on these queues such that the recursion is terminated by solving with the original MoM models with $M-B$ queues. In this case, we have found that the computational costs of the choices $B=1$ and $B=M$ are always better than these intermediate cases, unless $M-B=1$. Yet, in this more favorable cases, the costs of the intermediate choice of $B$ have the same order of magnitude of the best between $B=1$ and $B=M$, therefore the savings of these intermediate cases seem marginal and do not motivate a specialized implementation of the algorithm. As a result, we believe that the multi-branched recursion approach is best implemented with a choice $B=M$ which provides the biggest savings with respect to the original MoM on the largest number of choices of $M$ and $R$.

%\begin{figure}[t]
% \centering
%    \includegraphics[width=\textwidth]{timerequirements-singlestep.eps}\\
%  \caption{\footnotesize Comparison of divide-and-conquer MoM, MVA, and RECAL for different values of $M$, $R$, and $N$.}
%  \label{FIG:timerequirements}
%\end{figure}
%
%\begin{figure}[t]
% \centering
%    \includegraphics[width=\textwidth]{timerequirements-singlestep.eps}\\
%  \caption{\footnotesize Chart illustrating the best performing algorithm among MoM, divide-and-conquer MoM, MVA, and RECAL for different values of $M$, $R$, and $N$.}
%  \label{FIG:timerequirements}
%\end{figure}
%\section{Mean Value Analysis}

{\em Comparison with MVA Algorithm}. As a final remark, regardless of the branching level $B$ used, the computation of $\vec V(\vec N)$ from $\vec V(\vec 0)$ has an $O(N^2\log N)$ time complexity and an $O(N\log N)$ space complexity as the total population $N$ grows. Since MVA is $O(N^R)$ in time and space complexities, it is immediately clear that for sufficiently large populations MoM is always faster and less memory consuming than MVA. Savings are obtained by MoM already for populations composed by few tens of jobs\cite{Cas06b}. Therefore, since the original MoM is already much more scalable than MVA, it is an immediate consequence that the generalized MoM with GCEs, which always performs better than MoM, will be always several orders of magnitude more efficient than MVA or other methods such as RECAL or LBANC. We point to \cite{Cas06b} for a comparison of the original MoM with these methods supporting the statements in this subsection.

%\section{Computing State-Dependent Metrics}

\section{Conclusions}
\label{SEC:conclusions}
In this paper, we have presented a generalization of the Method of Moments (MoM), a recently proposed algorithm for the exact analysis of multiclass queueing network models which are widely used in capacity planning of computer systems and networks \cite{Cas06b,Cas09}. We have integrated in the MoM equations also the recursive formula used in the Convolution Algorithm~\cite{Buz73,ReiK75b}, here called the generalized convolution equation (GCE). We have shown that using the GCE in MoM significantly changes the structure of its recursion leading to the evaluation of models with different number of queues and which can be solved much more efficiently than the larger models considered by MoM. As a result, the computational costs in time and space of the generalized algorithm are several orders of magnitude smaller than the original MoM recursion.

As a possible extension of this work, we believe that the Convolution Algorithm equation considered in this paper could benefit also the Class-Oriented Method of Moments (CoMoM) algorithm presented in \cite{Cas09}. This algorithm can be seen as the dual of the MoM algorithm, where the basis of normalizing constants considered in the recursion is defined in such a way that a different tradeoff between number of queues and classes is considered and this favors the solution of models with many classes compared to the original MoM. The generalization of CoMoM with GCEs could possibly further enhance its scalability on models with many classes.

\bibliographystyle{plain}
\footnotesize
%\bibliography{../casale}

\begin{thebibliography}{}

\end{thebibliography}


\begin{thebibliography}{10}

\bibitem{Bar79}
Y.~Bard.
\newblock Some extensions to multiclass queueing network analysis.
\newblock In M.~Arato, A.~Butrimenko, and E.~Gelenbe, editors, {\em Proc. of
  the 3rd Int'l Symp. on Model. and Performance Evaluation of Comp. Syst.},
  pages 51--62, 1979.

\bibitem{BasCMP75}
F.~Baskett, K.~M. Chandy, R.~R. Muntz, and F.~G. Palacios.
\newblock Open, closed, and mixed networks of queues with different classes of
  customers.
\newblock {\em JACM}, 22(2):248--260, 1975.

\bibitem{BerM93}
A.~Bertozzi and J.~McKenna.
\newblock Multidimensional residues, generating functions, and their
  application to queueing networks.
\newblock {\em SIAM Review}, 35(2):239--268, 1993.

\bibitem{BolGMT98}
G.~Bolch, S.~Greiner, H.~de~Meer, and K.~S. Trivedi.
\newblock {\em Queueing Networks and Markov Chains}.
\newblock Wiley and Sons, 1998.

\bibitem{BruB80}
S.~C. Bruell and G.~Balbo.
\newblock {\em Computational Algorithms for Closed Queueing Networks}.
\newblock North-Holland, 1980.

\bibitem{Buz73}
J.~P. Buzen.
\newblock Computational algorithms for closed queueing networks with
  exponential servers.
\newblock {\em Comm. of the ACM}, 16(9):527--531, 1973.

\bibitem{Cas06b}
G.~Casale.
\newblock An efficient algorithm for the exact analysis of multiclass queueing
  networks with large population sizes.
\newblock In {\em Proc. of joint ACM SIGMETRICS/IFIP Performance}, pages
  169--180. ACM Press, 2006.

\bibitem{Cas09}
G.~Casale.
\newblock {CoMoM}: Efficient class-oriented evaluation of multiclass
  performance models.
\newblock {\em IEEE Trans. on Software Engineering}, to appear in 2009.

\bibitem{ChaN82}
K.~M. Chandy and D.~Neuse.
\newblock Linearizer: {A} heuristic algorithm for queuing network models of
  computing systems.
\newblock {\em Comm. of the ACM}, 25(2):126--134, 1982.

\bibitem{ChaS80}
K.~M. Chandy and C.~H. Sauer.
\newblock Computational algorithms for product-form queueing networks models of
  computing systems.
\newblock {\em Comm. of the ACM}, 23(10):573--583, 1980.

\bibitem{ChoLW95}
G.~L. Choudhury, K.~K. Leung, and W.~Whitt.
\newblock Calculating normalization constants of closed queuing networks by
  numerically inverting their generating functions.
\newblock {\em JACM}, 42(5):935--970, 1995.

\bibitem{Coh78}
D.~J.~A. Cohen.
\newblock {\em Basic Techniques of Combinatorial Theory}.
\newblock John Wiley and Sons, 1978.

\bibitem{ConG86}
A.~E. Conway and N.~D. Georganas.
\newblock {\sc RECAL} - {A} new efficient algorithm for the exact analysis of
  multiple-chain closed queueing networks.
\newblock {\em JACM}, 33(4):768--791, 1986.

\bibitem{CreSS02}
P.~Cremonesi, P.~J. Schweitzer, and G.~Serazzi.
\newblock A unifying framework for the approximate solution of closed
  multiclass queuing networks.
\newblock {\em IEEE Trans. on Computers}, 51:1423--1434, 2002.

\bibitem{DenB78}
P.~J. Denning and J.~P. Buzen.
\newblock The operational analysis of queueing network models.
\newblock {\em ACM Computing Surveys}, 10(3):225--261, 1978.

\bibitem{GorN67}
W.~J. Gordon and G.~F. Newell.
\newblock Closed queueing systems with exponential servers.
\newblock {\em Oper. Res.}, 15(2):254--265, 1967.

\bibitem{HarC02}
P.~G. Harrison and S.~Coury.
\newblock On the asymptotic behaviour of closed multiclass queueing networks.
\newblock {\em Performance Evaluation}, 47(2):131--138, 2002.

\bibitem{KouB03}
S.~Kounev and A.~Buchmann.
\newblock Performance modeling and evaluation of large-scale j2ee applications.
\newblock In {\em Proc. of CMG Conference}, pages 273--283, 2003.

\bibitem{Lam82}
S.~Lam.
\newblock Dynamic scaling and growth behavior of queueing network normalization
  constants.
\newblock {\em JACM}, 29(2):492--513, 1982.

\bibitem{Lam83}
S.~Lam.
\newblock A simple derivation of the {\sc mva} and {\sc lbanc} algorithms from
  the convolution algorithm.
\newblock {\em IEEE Trans. on Computers}, 32:1062--1064, 1983.

\bibitem{MitM85}
D.~Mitra and J.~McKenna.
\newblock Asymptotic expansions for closed markovian networks with
  state-dependent service rates.
\newblock {\em JACM}, 33(3):568--592, July 1985.

\bibitem{ReiK75b}
M.~Reiser and H.~Kobayashi.
\newblock Queueing networks with multiple closed chains: Theory and
  computational algorithms.
\newblock {\em IBM J. Res. Dev.}, 19(3):283--294, 1975.

\bibitem{ReiL80}
M.~Reiser and S.~S. Lavenberg.
\newblock Mean-value analysis of closed multichain queueing networks.
\newblock {\em JACM}, 27(2):312--322, 1980.

\bibitem{Sch79}
P.~J. Schweitzer.
\newblock Approximate analysis of multiclass closed networks of queues.
\newblock In {\em Proc. of the Int'l Conf. on Stoch. Control and Optim.}, pages
  25--29, Amsterdam, 1979.

\end{thebibliography}

\appendix

\end{document}